\documentclass[journal]{article}

\usepackage{graphicx,url}
\usepackage{dcolumn}
\usepackage{bm}
\usepackage{amsmath,amsfonts,amssymb}

\newtheorem{definition}{\noindent{\it Definition}}[section]
\newtheorem{theorem}{\noindent{\it Theorem}}[section]
\newtheorem{proposition}{\noindent{\it Proposition}}[section]
\newtheorem{lemma}{\noindent{\it Lemma}}[section]
\newtheorem{remark}{\noindent{\it Remark}}[section]
\newtheorem{corollary}{\noindent{\it Corollary}}[section]
\newtheorem{example}{\noindent{\it Example}}[section]
\newenvironment{proof}{\noindent{\it Proof:}}{$\hfill$ $\Box$\\ }

\begin{document}

\title{On cyclotomic cosets and code constructions}

\author{Giuliano G. La Guardia and Marcelo M. S. Alves
\thanks{Giuliano Gadioli La Guardia is with Department of Mathematics and Statistics,
State University of Ponta Grossa (UEPG), 84030-900, Ponta Grossa -
PR, Brazil. Marcelo M. S. Alves is with Department of Mathematics,
Federal University of Parana (UFPR), Av. Cel. Francisco H. dos
Santos, 210, Jardim das Americas, 81531-970, Curitiba-PR, Brazil.
Corresponding author: Giuliano G. La Guardia ({\tt \small
gguardia@uepg.br}). }}

\maketitle

\begin{abstract}
New properties of $q$-ary cyclotomic cosets modulo $n = q^{m} - 1$,
where $q \geq 3$ is a prime power, are investigated in this paper.
Based on these properties, the dimension as well as bounds for the
designed distance of some families of classical cyclic codes can be
computed. As an application, new families of nonbinary
Calderbank-Shor-Steane (CSS) quantum codes as well as new families
of convolutional codes are constructed in this work. These new CSS
codes have parameters better than the ones available in the
literature. The convolutional codes constructed here have free
distance greater than the ones available in the literature.
\end{abstract}

\emph{keywords}: cyclotomic cosets; BCH codes; CSS construction

\section{Introduction}

Properties of cyclotomic cosets are extensively investigated in the
literature in order to obtain the dimension as well as lower bounds
for the minimum distance of cyclic codes
\cite{Macwilliams:1977,Lidl:1997,Mandelbaum:1980,Sharma:2004,Yue:2000,YueHu:1996}.
Such properties were useful to derive efficient quantum codes
\cite{Calderbank:1998,Steane:1999,Ashikmin:2001,Xi:2004,Ketkar:2006,Salah:2006,Salah:2007,LaGuardia:2009,LaGuardia:2010,LaGuardia:2011,LaGuardia:2014}.
In \cite{Mandelbaum:1980}, the authors explored properties of binary
cyclotomic cosets to compute the ones containing two consecutive
integers. In \cite{Sharma:2004,Yue:2000,YueHu:1996}, properties of
$q$-ary cyclotomic cosets ($q$-cosets for short) modulo $q^{m} - 1$
were investigated. In \cite{Salah:2006,Salah:2007}, the authors
established properties on $q$-cosets modulo $n$, where $\gcd (n, q)
= 1$, to compute the exact dimension of BCH codes of small designed
distance, providing new families of quantum codes. Additionally,
they employed such properties to show necessary and sufficient
conditions for dual containing (Euclidean and Hermitian) BCH codes.
Recently, in \cite{LaGuardia:2009,LaGuardia:2011,LaGuardia:2014},
the author has investigated properties of $q$-cosets as well as
properties of $q^{2}$-cosets in order to construct several new
families of good quantum BCH codes.

Motivated by all these works, we show new properties of $q$-cosets
modulo $n = q^{m} - 1$, where $q\geq 3$ is a prime power. Since the
nonbinary case has received less attention in the literature, in
this paper we deal with nonbinary alphabets. As was said previously,
these properties allow us to compute the dimension and bounds for
the designed distance of some families of cyclic codes. Since the
true dimension and minimum distance of BCH codes are not known in
general, this paper contributes to this research. As an application
of these results, we construct families of new
Calderbank-Shor-Steane (CSS) quantum codes (i.e., CSS codes with new
parameters; codes with parameters not known in the literature) as
well as new families of convolutional codes. These new CSS codes
have parameters given by
\begin{itemize}
\item  $[[q^{2}-1, q^{2} - 4c + 5, d \geq c ]]_{q}$,\\
where $2 \leq c \leq q$ and $q\geq 3$ is a prime power;

\item $[[n, n- 2m(c-2)-m/2-1, d \geq c]]_{q},$\\ where $n = q^{m} -1$, $q \geq 3$
is a prime power, $2\leq c\leq q$ and $m\geq 2$ is even;

\item $[[n, n - m(2c - 3)-1, d \geq c]]_{q}$,\\ where $n = q^{m} -1$, $q \geq 3$ is
a prime power, $m\geq 2$ and $2 \leq c\leq q$.
\end{itemize}

The new convolutional codes constructed here have parameters
\begin{itemize}
\item $(n, n-2q+1, 2q - 3; 1,
d_{free} \geq 2q +1)_{q}$,\\
where $q\geq 4$ is a prime power and $n= q^{2} -1$;

\item $(n, n-2q, 2q - 4;
1, d_{free} \geq 2q +1)_{q}$,\\
where $q\geq 4$ is a prime power and $n= q^{2} -1$;

\item $(n, n-2[q+i],
2[q-2-i]; 1, d_{free} \geq 2q+1)_{q}$, where $1\leq i\leq q - 3$,
$q\geq 4$ is a prime power and $n= q^{2} -1$;

\item $(n, n-2q+1, 1; 1, d_{free} \geq q+2)_{q}$,\\
$q\geq 4$ is a prime power and $n= q^{2} -1$;

\item $(n, n-2q+1, 2i+1; 1, d_{free} \geq q+i+3)_{q}$,\\
$1\leq i\leq q - 3$, $q\geq 4$ is a prime power and $n= q^{2} -1$.
\end{itemize}

The paper is organized as follows. In Section~\ref{sec2}, we review
some basic concepts on $q$-cosets and cyclic codes. In
Section~\ref{sec3}, we present new results and properties of
$q$-cosets. In Section~\ref{sec4}, by applying some properties of
$q$-cosets developed in the previous section, we compute the
dimension and lower bounds for the minimum distance of some families
of classical cyclic codes. Further, we utilize these cyclic codes to
construct new good quantum codes by applying the CSS construction.
In Section~\ref{sec5}, we utilize the classical cyclic codes
constructed in Section~\ref{sec4} to derive new families of
convolutional codes with greater free distance. Section~\ref{sec6}
is devoted to compare the new code parameters with the ones
available in the literature and, in Section~\ref{sec7}, a summary of
the paper is given.

\section{Background}\label{sec2}

In this section, we review the basic concepts utilized in this
paper. For more details we refer to \cite{Macwilliams:1977}.

\emph{Notation.} In this paper ${\mathbb Z}$ denotes the ring of
integers, $q \geq 3$ denotes a prime power, ${\mathbb F}_{q}$ is the
finite field with $q$ elements, $\alpha$ denotes a primitive element
of ${\mathbb F}_{q^{m}}$, $M^{(i)}(x)$ denotes the minimal
polynomial of ${\alpha}^i \in {\mathbb F}_{q^{m}}$ and $C^{\perp }$
denotes the Euclidean dual of a code $C$. We always assume that a
$q$-coset is considered modulo $n = q^{m} - 1$. The notation $m=
{{\operatorname{ord}}_{n}}(q)$ denotes the multiplicative order of $q$ modulo $n$.\\

Recall that a $q$-coset modulo $n = q^{m} - 1$ containing an element
$s$ is defined by ${\mathbb{C}}_{s} = \{s, sq, sq^{2},
sq^{3},\ldots, sq^{m_{s}-1} \},$ where $m_{s}$ is the smallest
positive integer such that $s{q}^{m_{s}} \equiv s \mod n$. In this
case, $s$ is the smallest positive integer of the coset. The
notation ${\mathbb{C}}_{[a]}$ means a $q$-coset containing $a$,
where $a$ is not necessarily the smallest integer in such coset.

\begin{theorem}\label{BCH}(BCH bound) Let $C$ be a cyclic code with generator
polynomial $g(x)$ such that, for some integers $b\geq 0,$
$\delta\geq 1$, and for $\alpha\in {\mathbb F}_{q^{m}},$ we have
$g({\alpha}^{b}) = g({\alpha}^{b+1})= \ldots =
g({\alpha}^{b+\delta-2})=0$, that is, the code has a sequence of
$\delta-1$ consecutive powers of $\alpha$ as zeros. Then the minimum
distance of $C$ is, at least, $\delta$.
\end{theorem}

A cyclic code of length $n$ over ${\mathbb F}_{q}$ is a BCH code
with designed distance $\delta$ if, for some integer $b\geq 0$, we
have $g(x)= \operatorname{l.c.m.}\{{M}^{(b)}(x),
{M}^{(b+1)}(x),\ldots$, ${M}^{(b+\delta-2)}(x)\}$, i.e., $g(x)$ is
the monic polynomial of smallest degree over ${\mathbb F}_{q}$
having ${{\alpha}^{b}}, {{\alpha}^{b+1}},\ldots,
{{\alpha}^{b+\delta-2}}$ as zeros. From the BCH bound, the minimum
distance of BCH codes is greater than or equal to their designed
distance $\delta$. If $C$ is a BCH code then a parity check matrix
is given by
\begin{eqnarray*}
H_{C} = \left[
\begin{array}{ccccc}
1 & {{\alpha}^{b}} & {{\alpha}^{2b}} & \cdots & {{\alpha}^{(n-1)b}} \\
1 & {{\alpha}^{(b+1)}} & {{\alpha}^{2(b+1)}} & \cdots & {{\alpha}^{(n-1)(b+1)}}\\
\vdots & \vdots & \vdots & \vdots & \vdots\\
1 & {{\alpha}^{(b+\delta-2)}} & \cdots & \cdots & {{\alpha}^{(n-1)(b+\delta-2)}}\\
\end{array}
\right],
\end{eqnarray*}
by expanding each entry as a column vector with respect to some
${\mathbb F}_{q}$-basis ${\mathcal B}$ of ${\mathbb F}_{q^{m}}$ over
${\mathbb F}_{q}$, where $m= {{\operatorname{ord}}_{n}}(q)$, after
removing any linearly dependent rows. The rows of the resulting
matrix over ${\mathbb F}_{q}$ are the parity checks satisfied by
$C$.

Let ${\mathcal B} =\{ b_{1}, \ldots, b_{m}\}$ be a basis of
${\mathbb F}_{q^{m}}$ over ${\mathbb F}_{q}$. If $u = (u_1,\ldots
,u_{n}) \in {\mathbb F}_{q^{m}}^{n}$ then one can write the vectors
$u_{i}$, $1\leq i\leq n$, as linear combinations of the elements of
${\mathcal B}$, i.e., $u_{i} = u_{i1}b_{1} +\ldots + u_{im}b_{m}$.
Let $u^{(j)} = (u_{1j},\ldots, u_{nj}) \in {\mathbb F}_{q}^{n}$,
where $1\leq j\leq m$. Then, if $v \in {\mathbb F}_{q}^{n}$, one has
$v\cdot u=0$ if and only if $v \cdot u^{(j)} = 0$ for all $1\leq
j\leq m$.

\section{New properties of cosets}\label{sec3}

In this section we explore the structure of $q$-cosets in order to
obtain new properties of them. As it is well known, the knowledge of
the structure (cardinality, disjoint cosets and so on) of cyclotomic
cosets provide us conditions to compute the dimension and the (lower
bounds for) minimum distance of (classical) cyclic codes. These two
parameters of cyclic codes are not known in general in the
literature. Therefore, we utilize our new properties of $q$-cosets
to compute these two parameters of some families of (classical)
cyclic codes. Further, we use these classical cyclic codes to
construct new quantum codes by applying the CSS construction.
Theorem~\ref{D} is the first result of this section.

\begin{theorem}\label{D}
Let $q$ be an odd prime power and ${\mathbb{C}}_{s}$ be a $q$-coset.
Then $s$ is even if and only if $\forall t \in {\mathbb{C}}_{s},$
$t$ is even.
\end{theorem}

\begin{proof}
Suppose first $s=2k$, where $k \in {\mathbb Z}$, and let $t$ be an
element of the coset ${\mathbb{C}}_{s}$ without considering the
modulo operation. Then $t = 2k{q}^l$, where $0\leq l\leq m_{s}-1$.
Applying the division with remainder for $t$ and $q^{m}-1$ one has
$2kq^{l} = (q^{m} - 1)a + r$, where $r \in {\mathbb Z}$, $0 \leq r <
q^{m} - 1$; so $r = 2k{q}^{l} -(q^{m} - 1)a$. Since $q^{m} - 1$ is
even, $r$ is also even.

Conversely, suppose that each $t$, where $t \in {\mathbb{C}}_{s}$,
(considering the modulo operation) is of the form $t=2k$, $k \in
{\mathbb Z}$. Then by applying again the division with remainder for
$sq^{l}$ and $q^{m}-1$ one obtains $sq^{l}= (q^{m} - 1)a + t$, where
$0 \leq t < q^{m} - 1$ is even. Since $t$ and $q^{m} - 1$ are even
also is $sq^{l}$, and because $q^{l}$ is odd it follows that $s$ is
even, as required. The proof is complete.
\end{proof}

As direct consequences of Theorem~\ref{D} we present straightforward
corollaries (Corollaries~\ref{DD} and \ref{DD1}).

\begin{corollary}\label{DD}
There are no consecutive integers belonging to the same $q$-coset
modulo $n=q^{m} - 1$, where $q$ is an odd prime power.
\end{corollary}

\begin{remark}
Note that in the binary case there exists at least one coset
containing two consecutive elements, namely ${\mathbb{C}}_1$.
\end{remark}

\begin{corollary}\label{DD1}
Suppose that ${\mathbb{C}}_{x}$ and ${\mathbb{C}}_{y}$ are two
$q$-cosets, where $q$ is an odd prime power. Assume also that $a \in
{\mathbb{C}}_{x}$ and $b \in {\mathbb{C}}_{y}$. If $a \not\equiv b
\mod 2$, then ${\mathbb{C}}_{x} \neq {\mathbb{C}}_{y}$.
\end{corollary}

In Theorem~\ref{DDD}, we introduce the positive integers $L_{s}$ in
order to compute the minimum absolute value of the difference
between elements in the same $q$-coset. This fact will be utilized
in the computation of the maximum designed distance of the
corresponding cyclic code (see Corollary~\ref{E}).

\begin{theorem}\label{DDD}
Let $q\geq 3$ be a prime power and ${\mathbb{C}}_{s}$ be a $q$-coset
with representative $s$. Define $L_{s} = \min \{ \mid [sq^{j}]_{n} -
[sq^{l}]_{n} \mid : 0 \leq j, l \leq m_{s}-1, j\neq l\}$, where
${\mid \cdot \mid}_{n}$ denotes the absolute value function and
$[\cdot]_{n}$ denotes the remainder modulo $n= q^{m} - 1$. Then one
has $L_{s} \geq q - 1$ for all $s$, where $s$ runs through the coset
representatives. Moreover, there exists at least one $q$-coset
${\mathbb{C}}_{s^{*}}$ such that $L_{s^{*}} = q - 1$.
\end{theorem}
\begin{proof}
Let ${\mathbb{C}}_{s}$ be a $q$-coset and assume that there exist
integers such that $0 \leq j, l\leq m_{s}-1$. Assume without loss of
generality that $l > j$. Applying the division with remainder for
$sq^{l}$ and $n$, $sq^{j}$ and $n$ one obtains $sq^{l}=an +
[sq^{l}]_{n}$ and $sq^{j}=bn + [sq^{j}]_{n}$, $a, b \in {\mathbb
Z}$, so $sq^{j}(q^{l-j}-1)= - sq^{j} = (a-b)n +
([sq^{l}]_{n}-[sq^{j}]_{n})$. Since $q-1$ divides $q^{l-j}-1$ and
also divides $n$, it follows that
$(q-1)|([sq^{l}]_{n}-[sq^{j}]_{n})$, hence $L_{s} \geq q - 1$. To
finish the proof, it suffices to consider the coset ${\mathbb
C}_{1}$ and its element $s^{*}=1 \in {\mathbb C}_{1}$; we have
$L_{s^{*}} = q - 1$. Now the result follows.
\end{proof}

\begin{corollary}\label{E}
Let $q \geq 3$ be a prime power. If $C$ is a $q$-ary cyclic code of
length $n=q^{m}-1$, $m \geq 2$, whose defining set ${\mathcal Z}$ is
the union of $c$ cosets ${\mathbb{C}}_{s+1}, {\mathbb{C}}_{s+2},
\ldots, {\mathbb{C}}_{s+c}$, where $c\geq 1$ is an integer, $s\geq
0$ is an integer and $1 \leq s+c \leq q-2$, then $\delta \leq c+2$,
where $\delta$ is the designed distance of $C$. In particular, if
${\mathcal Z}$ consists of only one $q$-coset then $\delta = 2$.
\end{corollary}
\begin{proof}
It suffices to note that ${\mathcal Z}$ contains at most a sequence
of $c+1$ consecutive integers and the result follows.
\end{proof}

Now, we define the concept of complementary $q$-coset and show some
interesting properties of it.

\begin{definition}\label{EE}
Let $q$ be a prime power. Given a $q$-coset ${\mathbb{C}}_{s}=\{s,
qs, q^{2}s, q^{3}s,$ $\ldots, q^{m_{s}-1}s \}$, a complementary
coset of \ ${\mathbb{C}}_{s}$ is a $q$-coset given by
${\mathbb{C}}_{r}=\{r, qr, q^{2}r, q^{3}r,$ $\ldots, q^{m_{r}-1}r
\}$, containing an element $q^{l}r$, where $0\leq l\leq m_{r}-1$,
such that $s + q^{l}r \equiv 0 \mod n $, where $n=q^{m} - 1$.
\end{definition}

Proposition~\ref{EEE} establishes some properties of complementary
$q$-cosets:

\begin{proposition}\label{EEE}
Let ${\mathbb{C}}_{s}=\{s, qs, q^{2}s, q^{3}s,\ldots, q^{m_{s}-1}s
\}$ be a $q$-coset modulo $n = q^{m}-1$. Then the following results
hold:

(i) For each given $q$-coset ${\mathbb{C}}_{s}$, there exists only
one complementary coset of ${\mathbb{C}}_{s}$, denoted by
${\overline{\mathbb{C}}}_{s}$.

(ii) The $q$-coset and its complementary coset have the same
cardinality.

(iii) Define the operation ${\mathbb{C}}_{s} \oplus
{\overline{\mathbb{C}}}_{r} = {\mathbb{C}}_{[s+q^{l}r]}$ ($q^{l}r$
is given in Definition~\ref{EE}); then one has ${\mathbb{C}}_{s}
\oplus {\overline{\mathbb{C}}}_{s} =\{0\}$.

(iv) If ${\mathbb{C}}_{r}$ is the complementary coset of
${\mathbb{C}}_{s}$ then $L_{s} = L_{r}$.

(v) ${\overline{{\overline{\mathbb{C}}}}}_{s} = {\mathbb{C}}_{s}$.
\end{proposition}
\begin{proof}
(i) Suppose that ${\mathbb{C}}_{s}$ is a $q$-coset and that
${\mathbb{C}}_{r_1}=\{r_1, qr_1, q^{2}r_1, q^{3}r_1,\ldots,$
$q^{m_{r_1}-1} r_1 \}$ and ${\mathbb{C}}_{r_2}=\{r_2, qr_2,
q^{2}r_2, q^{3}r_2,\ldots, q^{m_{r_2}-1}r_2 \}$ are two
complementary $q$-cosets of ${\mathbb{C}}_{s}$ with representatives
$r_1$ and $r_2$, respectively. From definition, there exist two
elements $q^{l}r_1$, $0\leq l\leq m_{r_1}-1$, and $q^{t} r_2$,
$0\leq t\leq m_{r_2}-1$ such that $s + q^{l}r_1 \equiv 0 \mod n$ and
$s + q^{t}r_2 \equiv 0 \mod n$. Thus $q^{l}r_1\equiv q^{t} r_2 \mod
n$, so ${\mathbb{C}}_{r_1}= {\mathbb{C}}_{r_2}$.

(ii) Let ${\mathbb{C}}_{s}$ be the coset containing $s$ with
cardinality $m_{s}$. Let ${\mathbb{C}}_{[n - s]}$ be the $q$-coset
containing $l = n - s$ of cardinality $m_l$, given by
${\mathbb{C}}_{[l]}= \{ n - s, (n - s)q, (n - s)q^{2}, \ldots, (n -
s)q^{m_{l}-1}\}$. It is clear that ${\mathbb{C}}_{[l]}$ is the
complementary coset of ${\mathbb{C}}_{s}$. We first prove that $(n -
s)q^{w}\not\equiv (n - s)q^{t}\mod n$ holds for each $0 \leq t,
w\leq m_{s}-1$, $t\neq w$, i.e., $m_{l} \geq m_{s}$. In fact,
seeking a contradiction, we assume that $(n - s)q^{w}\equiv (n -
s)q^{t}\mod n$ holds, where $0 \leq t, w\leq m_{s}-1$ and $t\neq w$.
Thus the congruence $s q^{w} \equiv s q^{t}\mod n$ holds, where $0
\leq t, w\leq m_{s}-1$, $t\neq w$, which is a contradiction. The
proof of the part $m_{l}\leq m_{s}$ is similar.

(iii) Straightforward.

(iv) Let ${\mathbb{C}}_{s}$ be a $q$-coset with complementary
${\mathbb{C}}_{r}$ and assume w.l.o.g. that $L_{s} = [sq^{t_1}]_{n}
- [sq^{t_2}]_{n}$, where $0 \leq t_1, t_2 \leq m_{s}-1, t_1 \neq
t_2$. Applying the division with remainder for $sq^{t_1}$ and $n$,
$sq^{t_2}$ and $n$, one obtains $sq^{t_1}=an + [sq^{t_1}]_{n}$ and
$sq^{t_2}=bn + [sq^{t_2}]_{n}$, $a, b \in {\mathbb Z}$, so $sq^{t_1}
- sq^{t_2}=(a-b)n + [sq^{t_1}]_{n}-[sq^{t_2}]_{n}$. From hypothesis,
there exists an integer $0\leq l\leq m_{r}-1$ such that $s \equiv
-rq^{l}\mod n$, so $rq^{(l+t_{2})}-rq^{(l+t_1)}\equiv
([sq^{t_1}]_{n}-[sq^{t_2}]_{n})\mod n$. We may assume w.l.o.g. that
$w_1 = l+t_{1}$, $w_{2} = l+t_{2}$ and $0 \leq w_1, w_{2} \leq
m_{r}-1, w_1 \neq w_2$, because, from Item (ii), $m_r =m_{s}$. Since
$L_{s} = [sq^{t_1}]_{n}-[sq^{t_2}]_{n}=\mid [rq^{w_{2}}]_{n}-
[rq^{w_1}]_{n}\mid$, it follows that $L_r \leq L_{s}$. The proof of
the part $L_{s} \leq L_r$ is similar.

(v) Straightforward.

The proof is complete.
\end{proof}



The following lemma gives us necessary and sufficient conditions
under which a cyclic code contains its Euclidean dual:

\begin{lemma}\cite[Lemma 1]{Salah:2007}\label{AAA}
Assume that $\gcd(q, n)=1$. A cyclic code of length $n$ over ${
\mathbb F}_q$ with defining set ${\mathcal Z}$ contains its
Euclidean dual code if and only if ${\mathcal Z}\cap {\mathcal
Z}^{-1} =\emptyset$, where ${\mathcal Z}^{-1}=\{-z \mod n : z \in
{\mathcal Z}\}$.
\end{lemma}

The next proposition characterizes Euclidean self-orthogonal cyclic
codes in terms of complementary $q$-cosets:

\begin{proposition}\label{EEEE}
Let $n = q^{m}-1$, where $q\geq 3$ is a prime power.  A cyclic codes
of length $n$ over ${\mathbb F}_q$ with defining set ${\mathcal
Z}=\displaystyle\cup_{i=1}^{l}{\mathbb{C}}_{r_{i}}$ contains its
Euclidean dual code if and only if
$[\displaystyle\cup_{i=1}^{l}{\overline{\mathbb{C}}}_{r_{i}}]\cap
{\mathcal Z}=\emptyset$.
\end{proposition}
\begin{proof}
Since the $q$-coset ${\mathbb{C}}_{-j}$ is the complementary coset
of ${\mathbb{C}}_{j}$, $j = r_1, r_2, \ldots, r_l$, the result
follows.
\end{proof}

Let us now consider the following result shown in \cite{Salah:2007}:

\begin{lemma}\cite[Lemmas 8 and 9]{Salah:2007}\label{F}
Let $n\geq 1$ be an integer and $q$ be a power of a prime such that
$\gcd(n,q)=1$ and ${q}^{\lfloor m/2 \rfloor} < n \leq q^{m}-1$,
where $m= {{\operatorname{ord}}_{n}}(q)$.

(a) The $q$-coset ${\mathbb{C}}_{x}= \{ xq^{j} \mod n : 0 \leq j < m
\}$ has cardinality $m$ for all $x$ in the range $ 1\leq x \leq n
{q}^{\lceil m/2 \rceil}/(q^{m}-1)$;

(b) If $x$ and $y$ are distinct integers in the range $1 \leq x,y
\leq \min \{ \lfloor n{q}^{\lceil m/2 \rceil }/(q^{m}-1)-1 \rfloor,
n-1 \}$ such that $x, y \not\equiv 0 \mod q$, then the $q$-cosets of
$x$ and $y$ (modulo $n$) are disjoint.
\end{lemma}

In Theorem~\ref{FFF}, the upper bound for the number of disjoint
$q$-cosets modulo $n = q^{m} - 1$ is improved, where $m$ is an even
integer. More specifically, if $m$ is even, we show that the number
of disjoint $q$-cosets is greater than the number of disjoint cosets
presented in Lemma~\ref{F}-Item (b). This fact is useful to compute
the dimension of BCH codes whose defining set contains such
$q$-cosets.

\begin{theorem}\label{FFF}
Let $n=q^{m}-1$, where $q$ is a prime power and $m$ is even. If $x$
and $y$ are distinct integers in the range $1 \leq x,y \leq
2{q}^{m/2}$, such that $x, y\not\equiv 0\mod q$, then the $q$-cosets
of $x$ and $y$ modulo $n$ are disjoint.
\end{theorem}

\begin{proof}
Recall the following result shown in (\cite[Theorem 2.3]{Yue:2000}):
\emph{Let $n=q^{m}-1$, where $q$ is a prime power and $m$ is even.
Let $s^{*} = \min \{t: t \in {\mathbb{C}}_{s} \}$ be the minimum
coset representative. If $0 \leq s \leq T$, where $T:= 2q^{m/2}$,
and $q \nmid s$ then $s = s^{*}$, and $T$ is the greatest value
having this property}.

From hypothesis, the inequalities $0 \leq s \leq T:=2{q}^{m/2}$
hold. Thus, for every $0 \leq x, y \leq T:=2{q}^{m/2}$ such that $x,
y\not\equiv 0\mod q$, it follows that the minimum coset
representatives for ${\mathbb{C}}_{x}$ and ${\mathbb{C}}_{y}$ are
$x$ and $y$, respectively. Since distinct minimum coset
representatives belong to disjoint $q$-cosets, ${\mathbb{C}}_{x}$
and ${\mathbb{C}}_{y}$ are disjoint, as required. We are done.
\end{proof}

\begin{lemma}\label{JJ}
Let $n = q^{m} - 1$, where $q\geq 3$ is a prime power and $c$ be a
positive integer. If the inequality $cq + 1 < \lfloor q^{\lceil
m/2\rceil}-1\rfloor$ holds then the $c$ $q$-cosets given by
${\mathbb{C}}_{q+1}, {\mathbb{C}}_{2q+1},
{\mathbb{C}}_{3q+1},\ldots, {\mathbb{C}}_{cq+1}$ are mutually
disjoint and each of them has $m$ elements. Moreover, each of them
are disjoint of the $q$-cosets ${\mathbb{C}}_{1}, {\mathbb{C}}_{2},
\ldots , {\mathbb{C}}_{c}$.
\end{lemma}
\begin{proof}
Apply Lemma~\ref{F}.
\end{proof}

Combining Lemma~\ref{JJ} and Theorem~\ref{JJJ} we can construct more
families of cyclic codes.

\begin{theorem}\label{JJJ}
Let $q \geq 3$ be a prime power and $n = q^{m} - 1$, with $m \geq
2$, and assume that $cq + 1 < \lfloor q^{\lceil
m/2\rceil}-1\rfloor$. Then the last elements in the $c$ cosets given
by $ {\mathbb{C}}_{q+1}, {\mathbb{C}}_{2q+1},
{\mathbb{C}}_{3q+1},\ldots, {\mathbb{C}}_{cq+1}$, form a sequence of
$c$ consecutive integers.
\end{theorem}

\begin{proof}
Consider the $c$ $q$-cosets given by $ {\mathbb{C}}_{q+1},
{\mathbb{C}}_{2q+1}, {\mathbb{C}}_{3q+1},\ldots,
{\mathbb{C}}_{cq+1}$. From Lemma~\ref{JJ}, these $q$-cosets have
cardinality $m$. Let ${\mathbb{C}}_{s}$ and ${\mathbb{C}}_{s+q}$ be
two of them. Let $u$ and $v$ be the last elements in
${\mathbb{C}}_{s}$ and ${\mathbb{C}}_{s+q}$, respectively, where $u$
and $v$ are integers considered without using the modulo $n$
operation. Let $t=m-1$; then $u = sq^{m-1} = sq^{t}$ and $v =
(s+q)q^{m-1} = (s+q)q^{t}$. Since $v = sq^{t}+ q^{t+1}$, it follows
that $v \equiv sq^{t}+1 \mod n$, i. e., $v \equiv u + 1 \mod n$.
Applying the division with remainder for $v$ and $n$ and for $u+1$
and $n$, there exist integers $a, b, r_1$ and $r_2$, where $0 \leq
r_1, r_2 < n$ such that $v = an + r_1$; $u + 1 = bn + r_2$. Since $v
\equiv u + 1\mod n$, it follows that $r_1 = r_2$. Since the
$q$-cosets $ {\mathbb{C}}_{q+1}, {\mathbb{C}}_{2q+1},
{\mathbb{C}}_{3q+1},\ldots, {\mathbb{C}}_{cq+1}$ have cardinality
$m\geq 2$, it follows that $r_1 = r_2\neq 0$. If $v^{*}=r_1 =
r_{2}$, one has $v = an + v^{*}$ and $u + 1 = bn + v^{*}$, where $1
\leq v^{*} < n$. Let $u^{*}$ be the remainder of $u$ modulo $n$.
Since $u = bn + v^{*} - 1$, where $0 \leq v^{*} - 1 < n$, it follows
that $v^{*} = u^{*} + 1$, as required. The proof is complete.
\end{proof}

\section{New quantum codes}\label{sec4}

In this section we apply some results of Section~\ref{sec3} in order
to construct CSS codes with parameters shown the Introduction. We
note that constructions of quantum codes derived from classical ones
by computing the generator or parity check matrices of the latter
codes, in several cases, does not provide families of codes but only
codes with specific parameters. This is one advantage of our
constructions presented here. Let us recall the well known CSS
quantum code construction:

\begin{lemma}\cite{Nielsen:2000,Calderbank:1998,Ketkar:2006}
Let $C_1$ and $C_2$ denote two classical linear codes with
parameters ${[n, k_1,d_1]}_{q}$ and ${[n,k_2,d_2]}_{q}$,
respectively, such that $C_2\subset C_1$. Then there exists an
${[[n, K = k_1- k_2, D]]}_{q}$ quantum code where $D = \min \{wt(c)
: c \in (C_1 \backslash C_2) \cup (\displaystyle C_{2}^{\perp}
\backslash \displaystyle C_{1}^{\perp}) \}$.
\end{lemma}

In order to proceed further we establish Lemma~\ref{good}. Note that
in Lemma~\ref{good} the structure and the cardinality of some
$q$-cosets are computed. These results allow us to compute the
dimension and lower bounds for the minimum distance of the
corresponding families of cyclic codes derived from such $q$-cosets.

\begin{lemma}\label{good}
Let $q \geq 3$ be a prime power and $n = q^{2} - 1$. Consider the
$(2q - 2)$ $q$-cosets modulo $n$ given by ${\mathbb{C}}_{0} =\{ 0\},
{\mathbb{C}}_{1}= \{ 1, \ q \}, {\mathbb{C}}_{2}= \{ 2, \ 2q \},
{\mathbb{C}}_{3}= \{ 3, \ 3q \}, \ldots, {\mathbb{C}}_{q-2}= \{ q-2,
\ (q-2)q \}, {\mathbb{C}}_{q+1} = \{ q+1\}, {\mathbb{C}}_{q+2}= \{
q+2, \ 1 + 2q \}, \ldots, {\mathbb{C}}_{2q-1}= \{ 2q-1, \ 1+(q-1)q
\}$. Then, these $q$-cosets are disjoint. In addition, with
exception of the $q$-cosets ${\mathbb{C}}_{0}$ and
${\mathbb{C}}_{q+1}$, that contain only one element, all of them
have exactly two elements.
\end{lemma}
\begin{proof}
It is easy to show that the inequalities $q^{2}-1 > 1+(q-1)q$ and
$q^{2}-1 > (q-2)q$ are true. It is clear that $q$-cosets
${\mathbb{C}}_{0}$ and ${\mathbb{C}}_{q+1}$ contain only one
element. Next we show that the remaining $q$-cosets have cardinality
two. If $l = lq$, where $2\leq l\leq q-1$ is an integer, since $l =
lq < q^{2}-1$ we obtain $q=1$, a contradiction since $q$ is a prime
power. Assume that $q + l = 1 + lq$, where $2\leq l\leq q-1$. Then
one has $l - 1 = q(l - 1)$. Since $q + l = 1 + lq < q^{2}-1$ and $l
- 1 \neq 0$, one obtains $q = 1$, a contradiction.

Since these $q$-cosets have different smallest representatives, it
follows that they are mutually disjoint.
\end{proof}

In Theorem~\ref{good1}, we construct new families of good nonbinary
CSS codes of length $q^2 - 1$.

\begin{theorem}\label{good1}
Let $q \geq 3$ be a prime power and let $n = q^{2} -1$. Then, there
exist quantum codes with parameters ${[[q^{2}-1, q^{2} - 4q + 5, d
\geq q]]}_{q}$.
\end{theorem}

\begin{proof}
Let $C_1$ be the classical BCH code generated by $g_1(x)$, that is
the product of the minimal polynomials $ M^{(0)}(x)M^{(1)}(x)\ldots
M^{(q-2)}(x)$, and let $C_2$ be the cyclic code generated by
$g_2(x)$, that is the product of the minimal polynomials $
M^{(i)}(x)$, where $M^{(i)}(x)$ are the minimal polynomials of
${\alpha}^i$ such that $i \notin \{ q + 1, q + 2, \ldots, 2q - 1
\}$. We know the minimum distance of the code $C_1$ is greater than
or equal to $q$ since its defining set contains the sequence of
$q-1$ consecutive integers given by $0, 1, \ldots, q-2$. From the
BCH bound, $C_1$ has minimum distance $d_1 \geq q$. Similarly, the
defining set of $C$ generated by the polynomial $h(x)= \frac{x^n -
1}{g_2(x)}$ contains the sequence of $q-1$ consecutive integers
given by $q+1, q+2, \ldots, 2q-1$ so, from the BCH bound, $C$ also
has minimum distance greater than or equal to $q$. Since the code
$\displaystyle C_{2}^{\perp}$ is equivalent to $C$, then it follows
that $\displaystyle C_{2}^{\perp}$ also has minimum distance greater
than or equal to $q$. Therefore, the resulting CSS code has minimum
distance $d\geq q$.

We know the defining set ${\mathcal Z}_1$ of $C_1$ has $q -1$
disjoint $q$-cosets. Moreover, from Lemma~\ref{good}, all of them
(except coset ${\mathbb{C}}_{0}$) have two elements. Thus, $C_1$ has
dimension $k_1 = q^{2}- 2q + 2$. Similarly, the dimension of $C_2$
equals $k_2 = 2q - 3$, so $k_1 - k_2 = q^{2} -4q + 5$. Applying the
CSS construction to the codes $C_1$ and $C_2$, we can get a CSS code
with parameters ${[[q^{2}-1, q^{2} - 4q + 5, d \geq q]]}_{q}$. The
proof is complete.
\end{proof}

We illustrate Theorem~\ref{good1} by means of a graphical scheme:

\begin{eqnarray*}
\underbrace{ \overbrace{ {\mathbb{C}}_0 {\mathbb{C}}_1 \
{\mathbb{C}}_2 \ \ldots \ {\mathbb{C}}_{q-2}}^{C_1}}_{C_2} \\
\overbrace{{\mathbb{C}}_{q+1} \ {\mathbb{C}}_{q+2} \ldots \
{\mathbb{C}}_{2q-1}}^{C}
\ \underbrace{{\mathbb{C}}_{r_1} \ldots {\mathbb{C}}_{r_n}}_{C_2}.\\
\end{eqnarray*}

The union of the $q$-cosets ${\mathbb{C}}_0, {\mathbb{C}}_1, \ldots,
{\mathbb{C}}_{q-2}$ is the defining set of code $C_1$; the union of
the $q$-cosets ${\mathbb{C}}_0, {\mathbb{C}}_1, \ldots,
{\mathbb{C}}_{q-2}, {\mathbb{C}}_{r_1}, \ldots, {\mathbb{C}}_{r_j}$
is the defining set of $C_2$, where ${\mathbb{C}}_{r_1}, \ldots,
{\mathbb{C}}_{r_j}$ are the remaining $q$-cosets in order to
complete the set of all $q$-cosets; and the union of the $q$-cosets
${\mathbb{C}}_{q+1}, {\mathbb{C}}_{q+2}, \ldots, {\mathbb{C}}_{2q -
1}$ is the defining set of $C$.

Proceeding similarly as in the proof of Theorem~\ref{good1}, we can
also generate new families of quantum codes by means of Corollary~\ref{good2}:

\begin{corollary}\label{good2}
There exist quantum codes with parameters $[[q^{2}-1, q^{2} - 4c +
5,$ $d \geq c]]_{q}$, where $c < q$, and $q \geq 3$ is a prime
power.
\end{corollary}
\begin{proof}
Let $C_{1}$ and $C_{2}$, respectively, be the BCH codes generated by
the product of the minimal polynomials $C_{1}=\langle
M^{(0)}(x)M^{(1)}(x)M^{(2)}(x)\ldots M^{(c-2)}(x)\rangle$ and
$C_{2}=\langle \displaystyle \prod_{i} M^{(i)}(x)\rangle$, where
$M^{(i)}(x)$ are all minimal polynomials of ${\alpha}^i$ such that
$i \notin \{ q + 1, q + 2, \ldots, q+(c-1) \}$. Proceeding similarly
as in the proof of Theorem~\ref{good1}, new families of quantum
codes with good parameters ${[[q^{2}-1, q^{2} - 4c + 5, d \geq
c]]}_{q}$ are constructed. The proof is complete.
\end{proof}

\begin{example}
Applying Corollary~\ref{good2}, one can get quantum codes with
parameters ${[[15, 9, d \geq 3]]}_{4}$, ${[[15, 5, d \geq 4]]}_{4}$,
${[[24, 18, d \geq 3]]}_{5}$ and $[[24, 14$, $d \geq 4]]_{5}$.
\end{example}

Since we improved the upper bound for the number of disjoint
$q$-cosets (see Theorem~\ref{FFF}), we are able to construct new
families of CSS codes. Theorem~\ref{good3}, the main result of this
subsection, asserts the existence of such codes. Note that
Theorem~\ref{good1} is a particular case of Theorem~\ref{good3}.

\begin{theorem}\label{good3}
Let $n = q^{m} -1$, where $q \geq 3$ is a prime power and $m\geq 2$
is an even integer. Then there exist quantum codes whose parameters
are given by ${[[n, n- 2m(c-2)-m/2-1, d \geq c]]}_{q}$, where $2\leq
c\leq q$.
\end{theorem}

\begin{proof}
Recall the following result shown in \cite{YueHu:1996}\label{Y}:
\emph{$\mid {\mathbb{C}}_{s}\mid = m$ for all $0 < s < T:= 2q^{m/2}$
except $\mid {\mathbb{C}}_{q^{m/2}+1}\mid = m/2$ when $m$ is even}.

Let $C_1$ be the BCH code generated by the product of the minimal
polynomials $M^{(0)}(x)M^{(1)}(x)\ldots M^{(c-2)}(x)$ and $C_2$ be
the cyclic code generated by $g_2(x) = \displaystyle\prod_{i}
M^{(i)}(x)$, where $M^{(i)}(x)$ are the minimal polynomials of
${\alpha}^i$ such that $i \notin \{ q^{m/2}+ 1, q^{m/2}+2, \ldots,
q^{m/2}+c-1 \}$. From the BCH bound, the minimum distance $d_1$ of
$C_1$ satisfies $d_1\geq c$ since its defining set contains the
sequence $0, 1, \ldots, c-2$ of consecutive integers. Similarly, the
minimum distance of $\displaystyle C_{2}^{\perp}$ is also greater
than or equal to $c$, because $\displaystyle C_{2}^{\perp}$ is
equivalent to code $C =\langle (x^n - 1)/g_2 (x)\rangle$ and $C$
contains the sequence $q^{m/2}+ 1, q^{m/2}+2, \ldots, q^{m/2}+c-1$
of consecutive integers. The resulting CSS code has minimum distance
$d\geq c$. From construction we have $C_2 \subset C_1$. The
dimension of $C_1$ is given by $k_1 = n - m(c - 2) - 1$. Applying
Theorem~\ref{FFF}, since $q^{m/2}+c-1 < T:= 2q^{m/2}$ and because
the corresponding $q$-cosets are mutually disjoint, it follows that
$C_2$ has dimension $k_2 = m(c-2)+ m/2$; so $k_1 - k_2 = n -
2m(c-2)-m/2 -1$. Then there exists an ${[[n, n - 2m(c-2) -m/2 -1, d
\geq c]]}_{q}$ quantum code, as required.
\end{proof}

Applying Theorem~\ref{es}, given in the following, one can also
construct good quantum codes:

\begin{theorem}\label{es}
Let $n = q^{m} -1$, where $q \geq 3$ is a prime power and $m\geq 2$.
Then there exist quantum codes with parameters ${[[n, n - m(2c -
3)-1, d \geq c]]}_{q}$, where $2\leq c\leq q$ and $(c-1)q + 1 <
\lfloor q^{\lceil m/2\rceil}-1\rfloor$.
\end{theorem}
\begin{proof}
Let $C_1$ be the cyclic code generated by
$M^{(0)}(x)M^{(1)}(x)\ldots M^{(c-2)}(x)$, $2\leq c\leq q$, and
$C_2$ generated by $\displaystyle\prod_{i} M^{(i)}(x)$, where $ i
\notin \{ q + 1, 2q + 1, \ldots, (c - 1)q + 1 \}$. Applying
Lemma~\ref{JJ} and Theorem~\ref{JJJ}, and proceeding similarly as in
the proof of Theorem~\ref{good1} the result follows.
\end{proof}

\section{New convolutional codes}\label{sec5}

In this section, we apply the cyclic codes constructed in
Section~\ref{sec4} to derive new families of convolutional codes
with great free distance.

The theory of convolutional codes is well investigated in the
literature
\cite{Forney:1970,Lee:1976,Piret:1988,Johannesson:1999,Rosenthal:1999,York:1999,Rosenthal:2001,Schmale:2006,Gluesing:2006,LaGuardia:2012,LaGuardia:2013I,LaGuardia:2013II,LaGuardia:2014A}.
We assume the reader is familiar with the theory of convolutional
codes (see \cite{Johannesson:1999} for more details). Recall that a
polynomial encoder matrix $G(D)=(g_{ij}) \in { \mathbb
F}_{q}{[D]}^{k \times n}$ is called \emph{basic} if $G(D)$ has a
polynomial right inverse. A basic generator matrix is called
\emph{reduced} (or minimal \cite{Rosenthal:2001,Luerssen:2008}) if
the overall constraint length $\gamma =\displaystyle\sum_{i=1}^{k}
{\gamma}_i$, where ${\gamma}_i = {\max}_{1\leq j \leq n} \{ \deg
g_{ij} \}$, has the smallest value among all basic generator
matrices. In this case, the smallest overall constraint length
$\gamma$ is called the \emph{degree} of the code.

\begin{definition}\cite{Klapp:2007}
A rate $k/n$ convolutional code $C$ with parameters $(n, k, \gamma ;
\mu,$ $d_{f} {)}_{q}$ is a submodule of ${ \mathbb F}_q {[D]}^{n}$
generated by a reduced basic matrix $G(D)=(g_{ij}) \in { \mathbb
F}_q {[D]}^{k \times n}$, i.e., $C = \{ {\bf u}(D)G(D) | {\bf
u}(D)\in { \mathbb F}_{q} {[D]}^{k} \}$, where $n$ is the code
length, $k$ is the code dimension, $\gamma
=\displaystyle\sum_{i=1}^{k} {\gamma}_i$ is the \emph{degree}, $\mu
= {\max}_{1\leq i\leq k}\{{\gamma}_i\}$ is the \emph{memory} and
$d_{f}=$wt$(C)=\min \{wt({\bf v}(D)) : {\bf v}(D) \in C, {\bf
v}(D)\neq 0 \}$ is the \emph{free distance} of the code.
\end{definition}

Recall that the Euclidean inner product of two $n$-tuples ${\bf
u}(D) = {\sum}_i {\bf u}_i D^i$ and ${\bf v}(D) = {\sum}_j {\bf u}_j
D^j$ in ${\mathbb F}_q {[D]}^{n}$ is defined as $\langle {\bf
u}(D)\mid {\bf v}(D)\rangle = {\sum}_i {\bf u}_i \cdot {\bf v}_i$.
If $C$ is a convolutional code then we define its Euclidean dual
code as $C^{\perp }=\{ {\bf u}(D) \in {\mathbb F}_q {[D]}^{n} :
\langle {\bf u}(D)\mid {\bf v}(D)\rangle = 0$ for all ${\bf v}(D)\in
C\}$.

Let $C$ an ${[n, k, d]}_{q}$ block code with parity check matrix
$H$. We split $H$ into $\mu+1$ disjoint submatrices $H_i$ such that
$$H = \left[
\begin{array}{c}
H_0\\
H_1\\
\vdots\\
H_{\mu}\\
\end{array}
\right],$$ where each $H_i$ has $n$ columns, obtaining the
polynomial matrix $$G(D) = {\tilde H}_0 + {\tilde H}_1 D + {\tilde
H}_2 D^2 + \ldots + {\tilde H}_{\mu} D^{\mu},$$ where the matrices
${\tilde H}_i$, for all $1\leq i\leq \mu$, are derived from the
respective matrices $H_i$ by adding zero-rows at the bottom in such
a way that the matrix ${\tilde H}_i$ has $\kappa$ rows in total,
where $\kappa$ is the maximal number of rows among the matrices
$H_i$. The matrix $G(D)$ generates a convolutional code. Note that
$\mu$ is the memory of the resulting convolutional code generated by
$G(D)$. Let $\operatorname{rk}A$ denote the rank of the matrix $A$.

\begin{theorem}\cite[Theorem 3]{Aly:2007}\label{A}
Let $C \subseteq { \mathbb F}_{q}^{n}$ be an ${[n, k, d]}_{q}$
linear code with parity check $H \in { \mathbb F}_{q}^{(n-k)\times
n}$, partitioned into submatrices $H_0, H_1, \ldots, H_{\mu}$ as
above such that $\kappa = \operatorname{rk}H_0$ and
$\operatorname{rk}H_i \leq \kappa$ for $1 \leq i\leq \mu$. Let
$G(D)$ be the polynomial matrix given above.
Then the following conditions hold:\\
(a) The matrix $G(D)$ is a reduced basic generator matrix;\\
(b) Let $V$ be the convolutional code generated by $G(D)$ and
$V^{\perp}$ its Euclidean dual code. If $d_f$ and $d_f^{\perp}$
denote the free distances of $V$ and $V^{\perp}$, respectively,
$d_i$ denote the minimum distance of the code $C_i = \{ {\bf v}\in
{\mathbb F}_{q}^n : {\bf v} {\tilde H}_i^t =0 \}$ and $d^{\perp}$ is
the minimum distance of $C^{\perp}$, then one has $\min \{ d_0 +
d_{\mu} , d \} \leq d_f^{\perp} \leq  d$ and $d_f \geq d^{\perp}$.
\end{theorem}

In Theorem~\ref{mainconv}, the first result of this section, we
construct new convolutional codes:

\begin{theorem}\label{mainconv}
Assume that $q\geq 4$ is a prime power and $n= q^{2} -1$. Then there
exists a convolutional code with parameters $(n, n-2q+1, 2q - 3; 1,
d_{free} \geq 2q +1)_{q}$.
\end{theorem}

\begin{proof}
The $q$-coset ${\mathbb{C}}_{q-1}$ has two elements and it is
disjoint from all $q$-cosets given in Lemma~\ref{good}. Let $C$ be
the BCH code generated by $g(x)$, that is the product of the minimal
polynomials $ {M}^{(0)}(x) {M}^{(1)}(x) \cdots
{M}^{(q-2)}(x){M}^{(q-1)}(x) {M}^{(q+1)}(x)\cdots$
${M}^{(2q-1)}(x)$. A parity check matrix of $C$ is obtained from the
matrix
\begin{eqnarray*}
H = \left[
\begin{array}{ccccc}
1 & {{\alpha}^{(0)}} & {{\alpha}^{(0)}} & \cdots & {{\alpha}^{(0)}} \\
1 & {{\alpha}^{(1)}} & {{\alpha}^{(2)}} & \cdots & {{\alpha}^{(n-1)}}\\
\vdots & \vdots & \vdots & \vdots & \vdots\\
1 & {{\alpha}^{(q-1)}} & \cdots & \cdots & {{\alpha}^{(n-1)(q-1)}}\\
1 & {{\alpha}^{q+1}} & \cdots & \cdots & {{\alpha}^{(n-1)(q+1)}}\\
\vdots & \vdots & \vdots & \vdots & \vdots\\
1 & {{\alpha}^{(2q-1)}} & \cdots & \cdots & {{\alpha}^{(n-1)(2q-1)}}\\
\end{array}
\right],
\end{eqnarray*}
by expanding each entry as a column vector with respect to some
${\mathbb F}_{q}-$basis ${\mathcal B}$ of ${\mathbb F}_{q^{2}}$.
Note that since ${{\operatorname{ord}}_{n}}(q)=2$, each entry
contains $2$ rows. This new matrix $H_{C}$ is a parity check matrix
of $C$ and it has $4q-2$ rows. Since $C$ has dimension $k=n- \deg
g(x)$, i.e., $k = n - 4q + 4$, it follows that $H_{C}$ has rank
$4q-4$; $C$ has parameters $[n, n - 4q + 4, d\geq 2q+1]_{q}$.

We next assume that $C_0$ is the BCH code generated by ${M}^{(0)}(x)
{M}^{(1)}(x)\cdots$ ${M}^{(q-2)}(x){M}^{(q-1)}(x)$. $C_0$ has a
parity check matrix derived from the matrix
\begin{eqnarray*}
H_0 = \left[
\begin{array}{ccccc}
1 & {{\alpha}^{(0)}} & {{\alpha}^{(0)}} & \cdots & {{\alpha}^{(0)}}\\
1 & {{\alpha}^{(1)}} & {{\alpha}^{(2)}} & \cdots & {{\alpha}^{(n-1)}}\\
\vdots & \vdots & \vdots & \vdots & \vdots\\
1 & {{\alpha}^{(q-1)}} & \cdots & \cdots & {{\alpha}^{(n-1)(q-1)}}\\
\end{array}
\right],
\end{eqnarray*}
by expanding each entry as a $2$-column vector with respect to
${\mathcal B}$. This new matrix is denoted by $H_{C_0}$ (note that
$H_{C_0}$ is also a submatrix of $H_{C}$). The matrix $H_{C_0}$ has
rank $2q-1$ and the code $C_0$ has parameters $[n, n-2q +1, d_0 \geq
q+2]_{q}$.

Finally, let $C_1$ be the BCH code generated by ${M}^{(q+1)}(x)
{M}^{(q+2)}(x)\cdots$ \\ ${M}^{(2q-1)}(x)$. $C_1$ has parameters
${[n, n-2q + 3, d_1\geq q]}_{q}$. A parity check matrix $H_{C_1}$ of
$C_1$ is given by expanding each entry of the matrix
\begin{eqnarray*}
H_1 = \left[
\begin{array}{ccccc}
1 & {{\alpha}^{(q+1)}} & \cdots & \cdots & {{\alpha}^{(n-1)(q+1)}}\\
1 & {{\alpha}^{(q+2)}} & \cdots & \cdots & {{\alpha}^{(n-1)(q+2)}}\\
\vdots & \vdots & \vdots & \vdots & \vdots\\
1 & {{\alpha}^{(2q-1)}} & \cdots & \cdots & {{\alpha}^{(n-1)(2q-1)}}\\
\end{array}
\right],
\end{eqnarray*}
with respect to ${\mathcal B}$. Since $C_1$ has dimension $n-2q +
3$, $H_{C_{1}}$ has rank $2q - 3$ ($H_{C_{1}}$ is also a submatrix
of $H_{C}$).

We next construct a convolutional code $V$ generated by the matrix
$G(D)=\tilde H_{C_{0}}+ \tilde H_{C_{1}} D$, where $\tilde H_{C_{0}}
= H_{C_{0}}$ and $\tilde H_{C_{1}}$ is obtained from $H_{C_{1}}$ by
adding zero-rows at the bottom such that $\tilde H_{C_{1}}$ has the
number of rows of $H_{C_{0}}$ in total. According to Theorem~\ref{A}
Item (a), $G(D)$ is reduced and basic. We know that
$\operatorname{rk}H_{C_{0}}\geq\operatorname{rk}H_{C_{1}}$. By
construction, $V$ is a unit-memory convolutional code of dimension
$2q-1$ and degree ${\delta}_{V} = 2q-3$. The Euclidean dual
$V^{\perp}$ of the convolutional code $V$ has dimension $n-2q+1$ and
degree $2q-3$. From Theorem~\ref{A} Item (b), the free distance
$d_{f}^{\perp}$ of $V^{\perp}$ is bounded by $\min \{ d_0 + d_1 , d
\} \leq d_{f}^{\perp} \leq  d$, so $d_{f}^{\perp} \geq 2q +1$.
Hence, the convolutional code $V^{\perp}$ has parameters $(n,
n-2q+1, 2q - 3; 1, d_{f}^{\perp} \geq 2q +1)_{q}$. Now the result
follows.
\end{proof}

\begin{theorem}\label{mainconvB}
Let $q\geq 4$ be a prime power and $n= q^{2} -1$. Then there exists
an $(n, n-2q, 2q - 4; 1, d_{free} \geq 2q +1)_{q}$ convolutional
code.
\end{theorem}
\begin{proof}
Let $C$ be the BCH code generated by ${M}^{(0)}(x){M}^{(1)}(x)\cdots
{M}^{(q-2)}(x)$ ${M}^{(q-1)}(x) {M}^{(q+1)}(x)\cdots
{M}^{(2q-1)}(x)$, given in the proof of Theorem~\ref{mainconv}.
Suppose that $C_0$ is the BCH code generated by ${M}^{(0)}(x)
{M}^{(1)}(x)\cdots {M}^{(q-2)}(x)$ ${M}^{(q-1)}(x)$ ${M}^{(q+1)}(x)$
and assume that $C_1$ is the BCH code generated by
${M}^{(q+2)}(x)\cdots$ ${M}^{(2q-1)}(x)$. Proceeding similarly as in
the proof of Theorem~\ref{mainconv}, the result follows.
\end{proof}


\begin{theorem}\label{mainconvC}
Let $q\geq 4$ be a prime power and $n= q^{2} -1$. Then there exists
a convolutional code with parameters $(n, n-2[q+i], 2[q-2-i]; 1,
d_{free} \geq 2q+1)_{q}$, where $1\leq i\leq q - 3$.
\end{theorem}

\begin{proof}
Let $C$ be the $[n, n - 4q + 4, d\geq 2q+1]_{q}$ BCH code generated
by $ {M}^{(0)}(x) {M}^{(1)}(x) \cdot \ldots \cdot
{M}^{(q-2)}(x){M}^{(q-1)}(x) {M}^{(q+1)}(x)\cdot \ldots \cdot
{M}^{(2q-1)}(x)$, with parity check matrix $H_{C}$ of rank $4q-4$.
Suppose that $C_0$ is the BCH code generated by ${M}^{(0)}(x)
{M}^{(1)}(x)\cdot \ldots \cdot
{M}^{(q-2)}(x){M}^{(q-1)}(x){M}^{(q+1)}(x)\cdot \ldots \cdot
{M}^{(q+1+i)}(x)$, where $1\leq i\leq q - 3$, with parity check
matrix $H_{C_0}$ as per Theorem~\ref{mainconv}. $C_0$ has parameters
$[n, n-2q-2i, d_0 \geq q+i+3]_{q}$ and $H_{C_0}$ has rank $2(q+i)$.
Let $C_1$ be the BCH code generated by ${M}^{(q+2+i)}(x)\cdot \ldots
\cdot {M}^{(2q-1)}(x)$, with parity check matrix $H_{C_1}$ as per
Theorem~\ref{mainconv}. The code $C_1$ has parameters ${[n,
n-2(q-2-i), d_1\geq q-i-1]}_{q}$ and $H_{C_1}$ has rank $2(q-2-i)$.
The convolutional code $V$ generated by $G(D)=\tilde H_{C_{0}}+
\tilde H_{C_{1}} D$ is a unit-memory code of dimension $2(q+i)$, and
degree ${\delta}_{V} = 2(q-2-i)$. The dual $V^{\perp}$ of $V$ has
dimension $n-2(q+i)$, degree $2(q-2-i)$ and free distance
$d_{f}^{\perp} \geq 2q +1$. Therefore there exists a convolutional
code with parameters $(n, n-2[q+i], 2[q-2-i]; 1, d_{f}^{\perp} \geq
2q+1)_{q}$.
\end{proof}

\begin{theorem}\label{mainconvD}
Assume that $q\geq 4$ is a prime power and $n= q^{2} -1$. Then there
exists an $(n, n-2q+1, 2i+1; 1, d_{free} \geq q+i+3)_{q}$
convolutional code, where $1\leq i \leq q-3$.
\end{theorem}

\begin{proof}
Let $C$ be the BCH code generated by ${M}^{(0)}(x) {M}^{(1)}(x)
\cdot \ldots \cdot {M}^{(q-2)}(x)$ ${M}^{(q-1)}(x)
{M}^{(q+1)}(x){M}^{(q+2)}(x)\cdot \ldots \cdot {M}^{(q+1+i)}(x)$,
where $1\leq i \leq q-3$ with parity check matrix $H_{C}$ of rank
$2q+2i$; $C$ has parameters $[n, n - 2q-2i, d\geq q+3+i]_{q}$. Let
$C_0$ be the BCH code generated by ${M}^{(0)}(x) {M}^{(1)}(x)\cdot
\ldots \cdot {M}^{(q-2)}(x){M}^{(q-1)}(x)$. $C_0$ has parameters
$[n, n-2q +1, d_0 \geq q+2]_{q}$; $H_{C_0}$ has rank $2q-1$.
Finally, let $C_1$ be the BCH code generated by
${M}^{(q+1)}(x){M}^{(q+2)}(x)\cdot \ldots \cdot {M}^{(q+1+i)}(x)$.
$C_1$ has parameters ${[n, n-2i-1, d_1 \geq i+2]}_{q}$; $H_{C_1}$
has rank $2i+1$. Proceedings similarly as in the proof of
Theorem~\ref{mainconv}, the result follows.
\end{proof}


\begin{theorem}\label{mainconvE}
Assume that $q\geq 4$ is a prime power and $n= q^{2} -1$. Then there
exists a convolutional codes with parameters $(n, n-2q+1, 1; 1,
d_{free} \geq q+2)_{q}$.
\end{theorem}
\begin{proof}
Similar to that of Theorem~\ref{mainconv}.
\end{proof}

\begin{remark}
It is interesting to note that by applying the same construction
method shown in this section, several new families of convolutional
codes with other parameters can be constructed. Moreover, the
classical codes constructed in the proofs of Theorems~\ref{good3}
and \ref{es}, can be utilized similarly to derive novel families of
convolutional codes as well. These ideas can be explored in future
works.
\end{remark}

\section{Code Comparisons}\label{sec6}

In this section we compare the parameters of the new CSS codes with
the parameters of the best CSS codes shown in
\cite{Salah:2006,Salah:2007} and we also exhibit some new
(classical) convolutional codes constructed here.

The parameters ${[[n^{'}, k^{'}, d^{'}]]}_{q}={[[n,
n-2m(\lceil(\delta-1)(1-1/q)\rceil), d\geq\delta]]}_{q}$ displayed
in Tables~\ref{table1} and \ref{table2} are the parameters of the
codes shown in \cite{Salah:2006,Salah:2007}. In Table~\ref{table1},
the parameters ${[[n, k, d \geq c]]}_{q}$ assume the values
${[[q^{2}-1, q^{2} - 4c + 5, d \geq c ]]}_{q}$, where $2 \leq c \leq
q$ and $q\geq 3$ is a prime power.

As can be seen in Tables~\ref{table1} and \ref{table2}, the new CSS
codes have parameters better than the ones available in
\cite{Salah:2006,Salah:2007}. More precisely, fixing $n$ and $d$,
the new codes achieve greater values of the number of qudits than
the codes shown in \cite{Salah:2006,Salah:2007}.

\begin{table}[!pt]
\begin{center}
\caption{Code Comparisons \label{table1}}
\begin{tabular}{|c |c |}
\hline New CSS codes & Best CSS Codes in \cite{Salah:2006,Salah:2007}\\
\hline ${[[n, k, d \geq c]]}_{q}$ & ${[[n^{'}, k^{'}, d^{'}]]}_{q}$\\
\hline
\hline ${[[24, 18, d \geq 3]]}_{5}$ & ${[[24, 16, d^{'}\geq 3]]}_{4}$\\
\hline ${[[24, 10, d \geq 5]]}_{5}$ & ---\\
\hline ${[[48, 42, d \geq 3]]}_{7}$ & ${[[48, 40, d^{'}\geq 3]]}_{7}$\\
\hline ${[[48, 38, d \geq 4]]}_{7}$ & ${[[48, 36, d^{'}\geq 4]]}_{7}$\\
\hline ${[[48, 34, d \geq 5]]}_{7}$ & ${[[48, 32, d^{'}\geq 5]]}_{7}$\\
\hline ${[[48, 30, d \geq 6]]}_{7}$ & ${[[48, 28, d^{'}\geq 6]]}_{7}$\\
\hline ${[[48, 26, d \geq 7]]}_{7}$ & ---\\
\hline ${[[63, 57, d \geq 3]]}_{8}$ & ${[[63, 55, d^{'}\geq 3]]}_{8}$\\
\hline ${[[63, 53, d \geq 4]]}_{8}$ & ${[[63, 51, d^{'}\geq 4]]}_{8}$\\
\hline ${[[63, 49, d \geq 5]]}_{8}$ & ${[[63, 47, d^{'}\geq 5]]}_{8}$\\
\hline ${[[63, 45, d \geq 6]]}_{8}$ & ${[[63, 43, d^{'}\geq 6]]}_{8}$\\
\hline ${[[63, 41, d \geq 7]]}_{8}$ & ${[[63, 39, d^{'}\geq 7]]}_{8}$\\
\hline ${[[80, 54, d \geq 8]]}_{9}$ & ${[[80, 52, d^{'}\geq 8]]}_{9}$\\
\hline ${[[80, 50, d \geq 9]]}_{9}$ & ---\\
\hline ${[[120, 114, d \geq 3]]}_{11}$ & ${[[120, 112, d^{'}\geq 3]]}_{11}$\\
\hline ${[[120, 106, d \geq 5]]}_{11}$ & ${[[120, 104, d^{'}\geq 5]]}_{11}$\\
\hline ${[[120, 98, d \geq 7]]}_{11}$ & ${[[120, 96, d^{'}\geq 7]]}_{11}$\\
\hline ${[[120, 90, d \geq 9]]}_{11}$ & ${[[120, 88, d^{'}\geq 9]]}_{11}$\\
\hline ${[[120, 82, d \geq 11]]}_{11}$ & ---\\
\hline ${[[168, 162, d \geq 3]]}_{13}$ & ${[[168, 160, d^{'}\geq 3]]}_{13}$\\
\hline ${[[168, 154, d \geq 5]]}_{13}$ & ${[[168, 152, d^{'}\geq 5]]}_{13}$\\
\hline ${[[168, 146, d \geq 7]]}_{13}$ & ${[[168, 144, d^{'}\geq 7]]}_{13}$\\
\hline ${[[168, 138, d \geq 9]]}_{13}$ & ${[[168, 136, d^{'}\geq 9]]}_{13}$\\
\hline ${[[168, 130, d \geq 11]]}_{13}$ & ${[[168, 128, d^{'}\geq 11]]}_{13}$\\
\hline ${[[168, 122, d \geq 13]]}_{13}$ & ---\\
\hline
\end{tabular}
\end{center}
\end{table}

\begin{table}[!hpt]
\begin{center}
\caption{Code Comparisons
\label{table2}}
\begin{tabular}{|c | c|}
\hline New CSS codes & CSS Codes in \cite{Salah:2006,Salah:2007}\\
\hline ${[[n, n - 2m(c-2) -m/2 -1, d \geq c]]}_{q}$ & ${[[n^{'}, k^{'}, d^{'}]]}_{q}$\\
\hline ${[[15, 9, d \geq 3]]}_{4}$ & ${[[15, 7, d^{'} \geq 3]]}_{4}$\\
\hline ${[[15, 5, d \geq 4]]}_{4}$ & ${[[15, 3, d^{'}\geq 4]]}_{4}$\\
\hline ${[[24, 18, d \geq 3]]}_{5}$ & ${[[24, 16, d^{'} \geq 3]]}_{7}$\\
\hline ${[[24, 14, d \geq 4]]}_{5}$ & ${[[24, 12, d^{'} \geq 4]]}_{7}$\\
\hline ${[[24, 10, d \geq 5]]}_{5}$ & ${[[24, 8, d^{'} \geq 5]]}_{7}$\\
\hline ${[[63, 57, d \geq 3]]}_{8}$ & ${[[63, 55, d^{'} \geq 3]]}_{8}$\\
\hline ${[[63, 53, d \geq 4]]}_{8}$ & ${[[63, 51, d^{'} \geq 4]]}_{8}$\\
\hline ${[[63, 49, d \geq 5]]}_{8}$ & ${[[63, 47, d^{'} \geq 5]]}_{8}$\\
\hline ${[[63, 45, d \geq 6]]}_{8}$ & ${[[63, 43, d^{'} \geq 6]]}_{8}$\\
\hline ${[[63, 41, d \geq 7]]}_{8}$ & ${[[63, 39, d^{'} \geq 7]]}_{8}$\\
\hline ${[[63, 37, d \geq 8]]}_{8}$ & ${[[63, 35, d^{'} \geq 8]]}_{8}$\\
\hline ${[[255, 244, d \geq 3]]}_{4}$ & ${[[255, 239, d^{'} \geq 3]]}_{4}$\\
\hline ${[[255, 236, d \geq 4]]}_{4}$ & ${[[255, 231, d^{'} \geq 4]]}_{4}$\\
\hline ${[[624, 613, d \geq 3]]}_{5}$ & ${[[624, 608, d^{'} \geq 3]]}_{5}$\\
\hline ${[[624, 605, d \geq 4]]}_{5}$ & ${[[624, 600, d^{'} \geq 4]]}_{5}$\\
\hline ${[[624, 597, d \geq 5]]}_{5}$ & ${[[624, 592, d^{'} \geq 5]]}_{5}$\\
\hline
\hline ${[[n, n - m(2c - 3)-1, d \geq c]]}_{q}$, $2\leq c\leq q$ & ${[[n^{'}, k^{'}, d^{'}]]}_{q}$\\
\hline
\hline ${[[124, 102, d \geq 5]]}_{5}$ & ${[[124, 100, d^{'} \geq 5]]}_{5}$\\
\hline ${[[342, 320, d \geq 5]]}_{7}$ & ${[[342, 318, d^{'} \geq 5]]}_{7}$\\
\hline ${[[342, 314, d \geq 6]]}_{7}$ & ${[[342, 312, d^{'} \geq 6]]}_{7}$\\
\hline ${[[342, 308, d \geq 7]]}_{7}$ & ${[[342, 306, d^{'} \geq 7]]}_{7}$\\
\hline ${[[255, 242, d \geq 3]]}_{4}$ & ${[[255, 239, d^{'} \geq 3]]}_{4}$\\
\hline ${[[255, 234, d \geq 4]]}_{4}$ & ${[[255, 231, d^{'} \geq 4]]}_{4}$\\
\hline ${[[624, 611, d \geq 3]]}_{5}$ & ${[[624, 608, d^{'} \geq 3]]}_{5}$\\
\hline ${[[624, 603, d \geq 4]]}_{5}$ & ${[[624, 600, d^{'} \geq 4]]}_{5}$\\
\hline ${[[624, 595, d \geq 5]]}_{5}$ & ${[[624, 592, d^{'} \geq 5]]}_{5}$\\
\hline
\end{tabular}
\end{center}
\end{table}

Now, we address the comparison of the new convolutional codes with
the ones available in literature. The new convolutional codes
constructed here have great free distance. Note that the (classical)
convolutional codes constructed in
\cite{Aly:2007,LaGuardia:2012,LaGuardia:2013I} do not attain the
free distance of the codes constructed in the present paper.
Additionally, we did not have seen in literature convolutional codes
(having corresponding $n$ and $k$) with minimum distances as great
as the ones presented here. Because of this fact, it is difficult to
compare the new code parameters with the ones available in
literature. Therefore, we only exhibit, in Table~\ref{table3}, the
parameters of some convolutional codes constructed here.
\begin{table}[!hpt]
\begin{center}
\caption{New codes\label{table3}}
\begin{tabular}{|c |}
\hline  New convolutional codes\\
\hline $(n, n-2q+1, 2q-3; 1, d_{free}\geq 2q+1{)}_{q}$,\\ $n=q^{2}-1$, $q \geq 4$\\
\hline ${(15, 8, 5; 1, d_{free}\geq 9)}_{4}$\\
\hline ${(24, 15, 7; 1, d_{free}\geq 11)}_{5}$\\
\hline ${(48, 35, 11; 1, d_{free}\geq 15)}_{7}$\\
\hline ${(63, 48, 13; 1, d_{free}\geq 17)}_{8}$\\
\hline ${(80, 63, 15; 1, d_{free}\geq 19)}_{9}$\\
\hline ${(120, 99, 19; 1, d_{free}\geq 23)}_{11}$\\
\hline ${(168, 143, 23; 1, d_{free}\geq 27)}_{13}$\\
\hline ${(255, 224, 29; 1, d_{free}\geq 33)}_{16}$\\
\hline \hline $(n, n-2q, 2q - 4; 1, d_{free} \geq 2q
+1)_{q}$,\\ $n=q^{2}-1$, $q \geq 4$\\
\hline
\hline ${(15, 7, 4; 1, d_{free}\geq 9)}_{4}$\\
\hline ${(24, 14, 6; 1, d_{free}\geq 11)}_{5}$\\
\hline ${(120, 98, 18; 1, d_{free}\geq 23)}_{11}$\\
\hline ${(168, 142, 22; 1, d_{free}\geq 27)}_{13}$\\
\hline ${(255, 223, 28; 1, d_{free}\geq 33)}_{16}$\\
\hline \hline $(n, n-2(q+i), 2(q-2-i); 1, d_{free} \geq
2q+1)_{q}$,\\
$1\leq i\leq q - 3$,$n=q^{2}-1$, $q \geq 4$\\
\hline
\hline ${(15, 5, 2; 1, d_{free}\geq 9)}_{4}$\\
\hline ${(24, 12, 4; 1, d_{free}\geq 11)}_{5}$\\
\hline ${(24, 10, 2; 1, d_{free}\geq 11)}_{5}$\\
\hline ${(48, 32, 8; 1, d_{free}\geq 15)}_{7}$\\
\hline ${(48, 30, 6; 1, d_{free}\geq 15)}_{7}$\\
\hline ${(48, 28, 4; 1, d_{free}\geq 15)}_{7}$\\
\hline ${(48, 26, 2; 1, d_{free}\geq 15)}_{7}$\\
\hline ${(255, 221, 26; 1, d_{free}\geq 33)}_{16}$\\
\hline ${(255, 219, 24; 1, d_{free}\geq 33)}_{16}$\\
\hline ${(255, 213, 18; 1, d_{free}\geq 33)}_{16}$\\
\hline ${(255, 209, 14; 1, d_{free}\geq 33)}_{16}$\\
\hline ${(255, 203, 8; 1, d_{free}\geq 33)}_{16}$\\
\hline ${(255, 197, 2; 1, d_{free}\geq 33)}_{16}$\\
\hline
\hline $(n, n-2q+1, 2i+1; 1, d_{free} \geq q+i+3)_{q}$,\\
$1\leq i\leq q - 3$, $q\geq 4$ and $n= q^{2} -1$\\
\hline
\hline ${(15, 8, 3; 1, d_{free}\geq 8)}_{4}$\\
\hline ${(24, 15, 3; 1, d_{free}\geq 9)}_{5}$\\
\hline ${(24, 15, 5; 1, d_{free}\geq 10)}_{5}$\\
\hline ${(48, 35, 3; 1, d_{free}\geq 11)}_{7}$\\
\hline ${(48, 35, 5; 1, d_{free}\geq 12)}_{7}$\\
\hline ${(48, 35, 7; 1, d_{free}\geq 13)}_{7}$\\
\hline ${(48, 35, 9; 1, d_{free}\geq 14)}_{7}$\\
\hline
\end{tabular}
\end{center}
\end{table}

\section{Summary}\label{sec7}
In this paper we have shown new properties on $q$-cosets modulo $n =
q^{m} - 1$, where $q\geq 3$ is a prime power. Since the dimension
and minimum distance of BCH codes are not known, these properties
are important because they can be utilized to compute the dimension
and bounds for the designed distance of some families of cyclic
codes. Applying some of these properties, we have constructed
classical cyclic codes which were utilized in the algebraic
construction of new families of quantum codes by means of the CSS
construction. Additionally, new families of convolutional codes have
also been presented in this paper. These new quantum CSS codes have
parameters better than the ones available in the literature. The new
convolutional codes have free distance greater than the ones
available in the literature.

\section*{Acknowledgment}
I would like to thank the anonymous referee for his/her valuable
comments and suggestions that improve significantly the quality and
the presentation of this paper. This research has been partially
supported by the Brazilian Agencies CAPES and CNPq.

\end{document}